\documentclass[
preprint, aps,pra,showpacs,amssymb,amsmath]{revtex4-1}
\usepackage{amsthm,graphicx}

\newcommand{\bi}{\begin{itemize}}
\newcommand{\ei}{\end{itemize}}
\newcommand{\be}{\begin{enumerate}}
\newcommand{\ee}{\end{enumerate}}
\newcommand{\bd}{\begin{description}}
\newcommand{\ed}{\end{description}}
\newcommand{\al}[1]{\begin{align} #1 \end{align}}

\theoremstyle{plain}
\newtheorem{lem}{Lemma}
\newtheorem{thm}[lem]{Theorem}
\newtheorem{mthm}{Main Theorem}

\newcommand{\ket}[1]{\ensuremath{| #1 \rangle}}

\newcommand{\ip}[2]{\ensuremath{\langle #1 | #2\rangle}}
\newcommand{\ceil}[1]{\ensuremath{\lceil #1 \rceil}}
\newcommand{\floor}[1]{\ensuremath{\lfloor #1 \rfloor}}

\newcommand{\dv}{\mid}

\begin{document}

\title{Exactness of the Original Grover Search Algorithm}

\author{Zijian Diao}
\affiliation{
Mathematics Department, Ohio University Eastern Campus, St
Clairsville, OH 43950, USA}
\email{diao@ohio.edu}

\begin{abstract}

It is well-known that when searching one out of four, the original
Grover's search algorithm is exact; that is, it succeeds with
certainty. It is natural to ask the inverse question: If we are not
searching one out of four, is Grover's algorithm definitely not
exact? In this article we give a complete answer to this question
through some rationality results of trigonometric functions.
\end{abstract}


\maketitle

\section{Introduction}\label{sec1}

Grover's algorithm \cite{Gro97} is one of the most significant
quantum algorithms \cite{Amb06}.  It provides a quadratic speedup
for the unsorted database search problem by amplifying the
probability amplitude of the search target.  When it was first
discovered, like most quantum algorithms, it was a probabilistic
algorithm; that is, it may fail with certain (albeit small)
probability. Currently, several schemes have been proposed to make
this algorithm exact, either by fine-tuning the amplitude
amplification operator \cite{BHT98, CK98, Hoy00, Lon01} or by
dynamical modification of the oracle function encoding the database
\cite{CD00}.  The study of exact quantum algorithms bears importance
in both practical applications and theoretical research of quantum
information science.

It is straightforward to verify that, when searching one target out
of a database of four entries, the original Grover's algorithm is
exact; that it, it succeeds with certainty.  Is this the only case
of exactness, excluding the trivial search of a database full of
search targets?  We provide a rigorous analysis to confirm this
conjecture in this article.  Reference~\cite{Hoy00} derives an
elegant phase condition for the amplitude amplification operator,
which is sufficient to ensure search with certainty. Unfortunately,
the phase shift $\pi$ in the original Grover's algorithm is exactly
what is ruled out in the assumption of this condition (cf.
\cite[Theorem~1]{Hoy00}). So the discussion there cannot be readily
applied.  Furthermore, our emphasis here deals with the opposite
direction to that used in \cite{Hoy00}. We fix the phase shift
($\pi$) first, then analyze whether the search is exact, under
varying initial success probability.

In the following sections we limit our discussion mostly to the
original Grover's algorithm, which searches for a single target. It
can be generalized in a straightforward fashion to the
multiple-target case \cite{CFC02, CS02} with the same essential
ingredients.  Similar arguments apply with minimal modification.

\section{Original Grover's Algorithm}

In this section we briefly review the procedure of the original
Grover's algorithm.  The problem dealt by the original Grover's
algorithm is as follows: Given an unsorted database containing $N$
items, $N\ge1$, how does one locate one particular target item?
Mathematically, the database is represented by an oracle function
$f(x)$, with $x\in
\{1, 2,
\dots, N\}$, defined by
\al{ f(x) = \begin{cases}
0 & \text{ if } x\ne w  \\
1 & \text{ if }x=w
\end{cases},}
where $w$ is the search target.  Grover's algorithm utilizes the
amplitude amplification operator $\mathcal{G}=
\mathcal{I}_s\mathcal{I}$, defined by
\begin{align}
  \mathcal{I}\ket{x} = (-1)^{f(x)}\ket{x},
\end{align}
or, equivalently, \al{\mathcal{I}= \mathbb{I}-2|{w}\rangle
\langle{w}|,} and  \al{
\mathcal{I}_s = 2|s\rangle \langle s|-\mathbb{I},}
where $|s\rangle =\frac{1}{\sqrt{N}}\Big(\sum_{x=1}^{N}|x\rangle
\Big)$, the uniform superposition (the average) of all possible
basis states, and $\mathbb{I}$ is the identity operator.
$\mathcal{I}$ is the selective sign-flipping operator, which
selectively flips the sign of the target state $\ket{w}$.
$\mathcal{I}_s$ is the inversion around the average operator, which
reflects a given state vector around $\ket{s}$.

The procedure of Grover's algorithm is as follows:
\begin{itemize}
\item[(1)]  prepare the initial state vector $|s\rangle$;

\item[(2)]  apply $\mathcal{G}$ on $|s\rangle $ for an appropriate
number of times (approximately $\frac{\pi}{4} \sqrt{N}$ times if $N$
is very large);

\item[(3)] measure the final state, which yields the target state
$\ket{w}$ with high probability.
\end{itemize}

The effect of the amplitude amplification operator, $\mathcal{G}$,
and why this algorithm works, can be best explained by a geometric
visualization (see Fig.~\ref{fig1}) on the plane spanned by
$\ket{s}$ and $\ket{w}$. When applied to a state vector $\ket{v}$,
the selective sign-flipping operator $\mathcal{I}$ flips the sign of
the component of $\ket{v}$ in the direction of $\ket{w}$, but leaves
all other components unchanged. So the pure effect is a reflection
of $\ket{v}$ about $\ket{w^\perp}$, the orthogonal vector to
$\ket{w}$. When applied to a state vector $\ket{v}$, the inversion
around the average operator $\mathcal{I}_s$ leaves the component in
the direction of $\ket{s}$ unchanged, but flips the signs of all the
other components. So the pure effect is a reflection of $\ket{v}$
about $\ket{s}$.  If we start from $\ket{s}$, one application of
$\mathcal{G}= \mathcal{I}_s\mathcal{I}$ reflects $\ket{s}$ first
about $\ket{w^\perp}$ and then about $\ket{s}$, hence rotates
$\ket{s}$ toward $\ket{w}$ by an angle of $2\theta$, where $\theta$
is the initial angle between $\ket{s}$ and $\ket{w^\perp}$ with
$\sin\theta = \cos(\frac{\pi}{2}-\theta)= \ip{s}{w} =
\frac{1}{\sqrt{N}}$.

\begin{figure}[htb]
\centering
\includegraphics{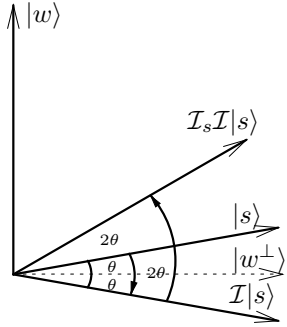}
\caption{Geometric visualization of Grover's algorithm.}
\label{fig1}
\end{figure}

It can be explicitly computed \cite[p.~252]{NC00} that, after $n$
iterations,
\al{\mathcal{G}^n\ket{s} = \sin((2n+1)\theta) \ket{w} +
\cos((2n+1)\theta) \ket{w^\perp}.}
So the success probability $p_n$ is $\sin^2((2n+1)\theta)$. When
$n=\frac{\pi}{4\theta}-\frac{1}{2}$, $(2n+1)\theta=\frac{\pi}{2}$,
and $p_n=1$. A measurement after $n$ steps yields $\ket{w}$ with
certainty.  However, $\frac{\pi}{4\theta}-\frac{1}{2}$ is not
necessarily an integer, so the optimal strategy is choosing $n$ to
be either $\ceil{\frac{\pi}{4\theta}-\frac{1}{2}}$ or
$\floor{\frac{\pi}{4\theta}-\frac{1}{2}}$ such that $(2n+1)\theta$
is the closest to $\frac{\pi}{2}$ in order to maximize $p_n$.  The
consequence is that $p_n$ is close, but not equal, to $1$, which
explains the probabilistic nature of the algorithm.

\section{Exactness of the Original Grover's Algorithm}

In this section we fully resolve the exactness of the original
Grover's algorithm.  Let us start from the special case of searching
one out of four.  Now $\sin\theta=\frac{1}{2}$, $\theta =
\frac{\pi}{6}$.  After one iteration, $p_1= \sin(3\theta)=
\sin\frac{\pi}{2}=1$.  We can find the target with certainty after
one oracle query (cf. Fig.~\ref{fig2}).  It is obvious that in order
for the algorithm to be exact, it is necessary for $\theta$ to be a
rational multiple of $\pi$.

\begin{figure}[htb]

\centering
\includegraphics{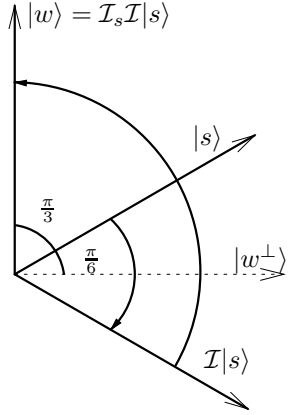}
\caption{Geometric visualization when searching one out of four.}
\label{fig2}

\end{figure}

The analysis in the rest of this section is motivated by
\cite[Chapter 4]{Rot98} and follows the same line of presentation.
Let us start from a basic result about the rational roots of
polynomials, adapted from \cite[Proposition 11, pp.~308]{DF04}.
First we define a polynomial to be monic if its leading coefficient
is $1$.
\begin{lem}\label{thm1}
Let $f(x)= x^n + a_{n-1}x^{n-1}+\dots + a_1 x + a_0$ be a monic
polynomial with integer coefficients. Then every rational root of
$f(x)$ is an integer.
\end{lem}
\begin{proof}
Suppose $x=\frac{A}{B}$, with $A$ and $B$ being relative prime and
$B>0$, is a rational root of $f(x)$.  Thus,
\al{\frac{A^n}{B^n}+a_{n-1}\frac{A^{n-1}}{B^{n-1}} + \dots +
a_1 \frac{A}{B} + a_0 &=0\\
A^n+ a_{n-1}A^{n-1}B + \dots + a_1AB^{n-1}+ a_0 B^n&=0 \\
B(a_{n-1}A^{n-1} + \dots + a_1AB^{n-2}+ a_0B^{n-1})&= -A^n
\label{eq1}}
From \eqref{eq1}, we have $B \dv A^n$, but $A$ and $B$ are
relatively prime, so $B=1$.  Therefore, $x=A$ is an integer.
\end{proof}

The following rationality result of trigonometric functions is
adapted from \cite{NZ72}.  We use $\mathbb{Q}$ to denote the set of
rational numbers.

\begin{lem} \label{lem1}
There exists a sequence of monic polynomials $f_n$ with integer
coefficients such that $f_n(2\cos\phi)=2\cos(n\phi)$, for all $n=1$,
$2$, \dots .
\end{lem}
\begin{proof} Let's construct this sequence of polynomials inductively
by $f_0(x) = 2$, $f_1(x) =x$, and $f_n(x)= xf_{n-1}(x)-f_{n-2}(x)$.
Clearly all $f_n$'s except $f_0$ are monic and all their
coefficients are integers. Also, $f_0$ and $f_1$ satisfy the cosine
property. Assume that $f_n(2\cos\phi)=2\cos(n\phi)$ for all indices
up to $n$. It is easy to verify that $f_{n+1}(2\cos\phi)=2\cos\phi\,
f_n(2\cos\phi)-f_{n-1}(2\cos\phi)=4\cos\phi\cos(n\phi)-2\cos((n-1)\phi)
=2\cos((n+1)\phi)$, which completes the induction proof.
\end{proof}

\begin{thm}
The only rational values for $\cos(r\pi)$ with $r\in \mathbb{Q}$ are
$0$, $\pm \frac{1}{2}$, and $\pm 1$.
\end{thm}
\begin{proof}
If $r\in \mathbb{Q}$, there exists a non-negative integer $n$ such
that $nr$ is an integer.  Let $f_n$ be the polynomial constructed in
Lemma~\ref{lem1}.  $f_n(2\cos(r\pi)) = 2\cos(nr\pi)= \pm 2$, so
$2\cos(r\pi)$ is a root of the polynomial $f_n(x)\pm 2$.
Lemma~\ref{thm1} tells us that if $2\cos(r\pi)$ is a rational
number, then $2\cos(r\pi)$ has to be an integer, that is, $0$, $\pm
1$, or $\pm 2$. Hence, the only rational values of $\cos(r\pi)$ are
$0$, $\pm\frac{1}{2}$, and $\pm 1$.
\end{proof}

Now we are in the position to prove our main result.
\begin{mthm}
Excluding the trivial search of a database full of search targets,
the original Grover's algorithm is exact if and only if searching
one out of four.
\end{mthm}
\begin{proof}
In order to succeed with certainty after a number of iterations, the
geometric interpretation of Grover's algorithm imposes the
restriction that the angle $\theta$ must be a rational multiple of
$\pi$, that is, of the form $r\pi$, where $r\in \mathbb{Q}$. On the
other hand, $\sin^2\theta = \frac{1}{N}$ ($\frac{t}{N}$ in the
multiple-target case, where $t$ is the number of targets) is a
rational number, and so is $\cos(2\theta) = 1-2\sin^2\theta=
1-\frac{2}{N}$ ($1-\frac{2t}{N}$ in the multiple-target case).
However, the only possible rational values of $\cos(2\theta)$ are
$0$, $\pm\frac{1}{2}$, and $\pm 1$, when $\theta=r\pi$, $r\in
\mathbb{Q}$. Let us analyze these five values one by one.

\be
\item When $\cos(2\theta)=1$, $\sin^2\theta=0$.  This is the trivial
search for a nonexisting target.

\item When $\cos(2\theta)=-1$, $\sin^2\theta=1$.  This is the
trivial search of a database where all the entries are targets.

\item When $\cos(2\theta)= 0$, $\sin^2\theta=\frac{1}{2}$, and
$\theta=\frac{\pi}{4}$.  The success probability after $n$ iteration
is $\sin^2((2n+1)\theta) = \sin^2\frac{(2n+1)\pi}{4}=\frac{1}{2}$,
which is never $1$.

\item When $\cos(2\theta)= -\frac{1}{2}$,
$\sin^2\theta=\frac{3}{4}$, and $\theta=\frac{\pi}{3}$.  The success
probability after $n$ iteration is $\sin^2(2n+1)\theta =
\sin^2\frac{(2n+1)\pi}{3}$, which is never $1$ ($0$ if
$3\dv 2n+1$ and $\frac{3}{4}$ if $3\nmid 2n+1$). \label{case3outof4}

\item When $\cos(2\theta)= \frac{1}{2}$, $\sin^2\theta=\frac{1}{4}$,
so $\theta=\frac{\pi}{6}$.  This is the familiar case of searching
one out of four.  One iteration yields the search target with
certainty.
\ee
Out of these, the exactness result in this theorem follows
naturally.
\end{proof}

As the final remark, if post-measurement processing is allowed,
there is one more special case where exactness can be achieved. When
there are three search targets in a database with four entries, the
success probability is $0$ after one iteration
(cf.~Case~\ref{case3outof4} in the proof of Main Theorem~1 with
$n=1$). If we measure at this point, we are bound to discover the
only nontarget in the database. To complete the search successfully,
choosing any of the other three entries will do. However, this
strategy can not be extended to similar scaled-up three out of four
cases. If there are more than one nontargets, we can determine and
rule out only one of them after the measurement. Choosing any of the
remaining entries does not necessarily yield a target anymore.

\section{Discussion}

We have rigorously shown that searching one out of four is the only
nontrivial case where the original Grover's algorithm is exact.  It
would be interesting to generalize the same kind of reasoning to the
generalized Grover's search with arbitrary phase shifts, in
particular the phase shifts of the form $r\pi$ with $r\in
\mathbb{Q}$, since they are easier to implement in practice. We
conjecture that a thorough analysis based on rationality
observations will provide us with similar results.

\begin{acknowledgments}

This work is supported by the Ohio University Research Committee
Award and Faculty Research Activity Fund.  The author also thanks
Peter H\o yer and Barry Sanders for helpful comments, and the
Institute for Quantum Information Science at University of Calgary
where part of the draft was prepared.

\end{acknowledgments}


\begin{thebibliography}{12}
\expandafter\ifx\csname natexlab\endcsname\relax\def\natexlab#1{#1}\fi
\expandafter\ifx\csname bibnamefont\endcsname\relax
  \def\bibnamefont#1{#1}\fi
\expandafter\ifx\csname bibfnamefont\endcsname\relax
  \def\bibfnamefont#1{#1}\fi
\expandafter\ifx\csname citenamefont\endcsname\relax
  \def\citenamefont#1{#1}\fi
\expandafter\ifx\csname url\endcsname\relax
  \def\url#1{\texttt{#1}}\fi
\expandafter\ifx\csname urlprefix\endcsname\relax\def\urlprefix{URL }\fi
\providecommand{\bibinfo}[2]{#2}
\providecommand{\eprint}[2][]{\url{#2}}

\bibitem[{\citenamefont{Grover}(1997)}]{Gro97}
\bibinfo{author}{\bibfnamefont{L.K.}~\bibnamefont{Grover}},
  \bibinfo{journal}{Phys.\ Rev.\ Lett.} \textbf{\bibinfo{volume}{79}},
  \bibinfo{pages}{325} (\bibinfo{year}{1997}).

\bibitem[{\citenamefont{Ambainis}(2006)}]{Amb06}
\bibinfo{author}{\bibfnamefont{A.}~\bibnamefont{Ambainis}},
  \bibinfo{note}{e-print arXiv: quant-ph/0609168v1}.

\bibitem[{\citenamefont{Brassard et~al.}(1998)\citenamefont{Brassard, H{\o}yer,
  and Tapp}}]{BHT98}
\bibinfo{author}{\bibfnamefont{G.}~\bibnamefont{Brassard}},
  \bibinfo{author}{\bibfnamefont{P.}~\bibnamefont{H{\o}yer}}, \bibnamefont{and}
  \bibinfo{author}{\bibfnamefont{A.}~\bibnamefont{Tapp}}, in
  \emph{\bibinfo{booktitle}{Proc.\ of the 25th Int.\ Colloquium on Automata,
  Languages and Programming}} (\bibinfo{publisher}{Springer},
  \bibinfo{address}{New York}, \bibinfo{year}{1998}), vol.
  \bibinfo{volume}{1443} of \emph{\bibinfo{series}{Lecture Notes in Comp.\
  Sci.}}, pp. \bibinfo{pages}{820--831}.

\bibitem[{\citenamefont{Chi and Kim}(1998)}]{CK98}
\bibinfo{author}{\bibfnamefont{D.-P.} \bibnamefont{Chi}} \bibnamefont{and}
  \bibinfo{author}{\bibfnamefont{J.}~\bibnamefont{Kim}}, in
  \emph{\bibinfo{booktitle}{Lecture at First NASA International Conference on
  Quantum Computing and Quantum Communications}} (\bibinfo{address}{Palm
  Springs, FL}, \bibinfo{year}{1998}).

\bibitem[{\citenamefont{H{\o}yer}(2000)}]{Hoy00}
\bibinfo{author}{\bibfnamefont{P.}~\bibnamefont{H{\o}yer}},
  \bibinfo{journal}{Physical Reviw A} \textbf{\bibinfo{volume}{62}},
  \bibinfo{pages}{052304} (\bibinfo{year}{2000}).

\bibitem[{\citenamefont{Long}(2001)}]{Lon01}
\bibinfo{author}{\bibfnamefont{G.-L.}~\bibnamefont{Long}},
  \bibinfo{journal}{Physical Review A} \textbf{\bibinfo{volume}{64}},
  \bibinfo{pages}{022307} (\bibinfo{year}{2001}).

\bibitem[{\citenamefont{Chen and Diao}(2000)}]{CD00}
\bibinfo{author}{\bibfnamefont{G.}~\bibnamefont{Chen}} \bibnamefont{and}
  \bibinfo{author}{\bibfnamefont{Z.}~\bibnamefont{Diao}}
  (\bibinfo{year}{2000}), \bibinfo{note}{e-print arXiv:quant-ph/0011109v3}.

\bibitem[{\citenamefont{Chen et~al.}(2002)\citenamefont{Chen, Fulling, and
  Chen}}]{CFC02}
\bibinfo{author}{\bibfnamefont{G.}~\bibnamefont{Chen}},
  \bibinfo{author}{\bibfnamefont{S.}~\bibnamefont{Fulling}}, \bibnamefont{and}
  \bibinfo{author}{\bibfnamefont{J.}~\bibnamefont{Chen}}, in
  \emph{\bibinfo{booktitle}{Mathematics of Quantum Computation}}, edited by
  \bibinfo{editor}{\bibfnamefont{R.}~\bibnamefont{Brylinski}} \bibnamefont{and}
  \bibinfo{editor}{\bibfnamefont{G.}~\bibnamefont{Chen}}
  (\bibinfo{publisher}{CRC Press}, \bibinfo{address}{Boca Raton, Florida},
  \bibinfo{year}{2002}), chap.~\bibinfo{chapter}{6}, pp.
  \bibinfo{pages}{135--160}.

\bibitem[{\citenamefont{Chen and Sun}(2002)}]{CS02}
\bibinfo{author}{\bibfnamefont{G.}~\bibnamefont{Chen}} \bibnamefont{and}
  \bibinfo{author}{\bibfnamefont{S.}~\bibnamefont{Sun}}, in
  \emph{\bibinfo{booktitle}{Mathematics of Quantum Computation}}, edited by
  \bibinfo{editor}{\bibfnamefont{R.}~\bibnamefont{Brylinski}} \bibnamefont{and}
  \bibinfo{editor}{\bibfnamefont{G.}~\bibnamefont{Chen}}
  (\bibinfo{publisher}{CRC Press}, \bibinfo{address}{Boca Raton, Florida},
  \bibinfo{year}{2002}), chap.~\bibinfo{chapter}{7}, pp.
  \bibinfo{pages}{161--168}.

\bibitem[{\citenamefont{Nielsen and Chuang}(2000)}]{NC00}
\bibinfo{author}{\bibfnamefont{M.}~\bibnamefont{Nielsen}} \bibnamefont{and}
  \bibinfo{author}{\bibfnamefont{I.}~\bibnamefont{Chuang}},
  \emph{\bibinfo{title}{Quantum Computation and Quantum Information}}
  (\bibinfo{publisher}{Cambridge University Press},
  \bibinfo{address}{Cambridge, U.K.}, \bibinfo{year}{2000}).

\bibitem[{\citenamefont{Rotman}(1998)}]{Rot98}
\bibinfo{author}{\bibfnamefont{J.}~\bibnamefont{Rotman}},
  \emph{\bibinfo{title}{Journey into Mathematics: An Introduction to Proofs}}
  (\bibinfo{publisher}{Prentice Hall}, \bibinfo{year}{1998}).

\bibitem[{\citenamefont{Dummit and Foote}(2004)}]{DF04}
\bibinfo{author}{\bibfnamefont{D.}~\bibnamefont{Dummit}} \bibnamefont{and}
  \bibinfo{author}{\bibfnamefont{R.}~\bibnamefont{Foote}},
  \emph{\bibinfo{title}{Abstract Algebra}} (\bibinfo{publisher}{Wiley},
  \bibinfo{year}{2004}), \bibinfo{edition}{3rd} ed.

\bibitem[{\citenamefont{Niven and Zuckerman}(1972)}]{NZ72}
\bibinfo{author}{\bibfnamefont{I.}~\bibnamefont{Niven}} \bibnamefont{and}
  \bibinfo{author}{\bibfnamefont{H.~S.} \bibnamefont{Zuckerman}},
  \emph{\bibinfo{title}{An Introduction to the Theory of Numbers}}
  (\bibinfo{publisher}{Wiley}, \bibinfo{year}{1972}).

\end{thebibliography}

\end{document}